\documentclass[a4paper]{llncs}

\usepackage{amsmath}
\usepackage{amssymb}
\usepackage[english]{babel}
\usepackage{color}
\usepackage{enumitem}
\usepackage{ellipsis}
\usepackage{hyperref}
\usepackage[misc,geometry]{ifsym} \usepackage{mathrsfs}
\usepackage{mathtools}

\let\emph\relax \DeclareTextFontCommand{\emph}{\bfseries\em}

\usepackage{tikz}
\usetikzlibrary{matrix}

\usepackage{breakurl}

\usepackage{todonotes}

\usepackage{lipsum}
 \renewenvironment{proof}{{\noindent \itshape Proof.}}{\qed\vspace{\baselineskip}}

\newcommand{\word}[1]{\ensuremath{\boldsymbol{#1}}}

\newcommand{\bfc}{\word{c}}
\newcommand{\bfe}{\word{e}}

\newcommand{\bfg}{\word{g}}
\newcommand{\bfk}{\word{k}}

\newcommand{\bfx}{\word{x}}
\newcommand{\bfy}{\word{y}}
\newcommand{\bfz}{\word{z}}

\newcommand{\mat}[1]{\ensuremath{\boldsymbol{#1}}}

\newcommand{\bfC}{\mat{C}}
\newcommand{\bfE}{\mat{E}}
\newcommand{\bfG}{\mat{G}}

\newcommand{\bfX}{\mat{X}}
\newcommand{\bfY}{\mat{Y}}

\newcommand{\bfGamma}{\mat{\Gamma}}

\renewcommand{\ker}{{\rm Ker}}

\newcommand{\N}{\mathbb{N}}
\newcommand{\F}[1]{\mathbb{F}_{#1}}
\newcommand{\K}{\mathbb{K}}

\newcommand{\Fq}{\mathbb{F}_q}
\newcommand{\Fqm}{\mathbb{F}_{q^m}}

\newcommand{\Fqmu}{\mathbb{F}_{q^{mu}}}

\newcommand{\code}[1]{\mathscr{#1}}
\newcommand{\BB}{\code{B}}
\newcommand{\CC}{\code{C}}

\newcommand{\Gab}[2]{\code{G}_{#1}(#2)}
\newcommand{\IGab}[3]{I\code{G}_{#1, #2}(#3)}

\newcommand{\ext}[2]{\textrm{Ext}_{#2}(#1)}
\newcommand{\Supp}{\textrm{Supp}}
\newcommand{\Rsupp}{\textrm{RowSupp}}

\newcommand{\Mspace}[2]{\mathcal{M}_{#1}(#2)}

\newcommand{\rk}[1][q]{\mathbf{rank}_{#1}}

\newcommand{\qpoly}{\mathcal{L}}
\newcommand\smvee{\hbox{$\scriptscriptstyle\vee$}}
\newcommand{\adj}[1]{{#1}^{\raise0.9ex\smvee}}

\newcommand{\qdeg}{\deg_q}

\newcommand{\eqdef}{\stackrel{\textrm{def}}{=}}
\renewcommand{\span}{\mathbf{Span}}
\renewcommand{\leq}{\leqslant}
\renewcommand{\le}{\leqslant}

\newcommand{\map}[4]{\left\{
    \begin{array}{ccc}
      #1 & \longrightarrow & #2\\
      #3 & \longmapsto & #4
    \end{array}\right.
}
\newcommand{\ie}{{\em i.e. }}
\newcommand{\eg}{{\em e.g. }}

\newcommand{\LIGA}{{\sc Liga}}

 \author{Maxime Bombar\inst{1,2}\and Alain Couvreur \inst{2,1}}

\institute{LIX, CNRS UMR 7161, \'Ecole Polytechnique,\\
  Institut Polytechnique de Paris,\\
1 rue Honor\'e d'Estienne d'Orves\\
91120 {\sc Palaiseau Cedex} \and
Inria\\
\email{\{maxime.bombar, alain.couvreur\}@inria.fr}
}
 
\begin{document}

\title{Right-hand side decoding of Gabidulin codes and applications}

\maketitle

\begin{abstract}
  We discuss the decoding of Gabidulin and interleaved Gabi\-dulin
  codes. We give the full presentation of a decoding algorithm for
  Gabidulin codes, which as Loidreau's seminal algorithm consists in
  localizing errors in the spirit of Berlekamp--Welch algorithm for
  Reed--Solomon codes. On the other hand this algorithm consists in
  acting on codewords on the right while Loidreau's algorithm considers
  an action on the left. This {\em right--hand side} decoder was
  already introduced by the authors in a previous work for
  cryptanalytic applications. We give a generalised version which
  applies to non--full length Gabidulin codes. Finally, we
  show that this algorithm turns out to provide a very clear and
  natural approach for the decoding of interleaved Gabidulin codes.
\end{abstract}
 \section*{Introduction}
Rank metric codes have been introduced in \cite{G85} by Gabidulin and
have found applications in cryptography
\cite{GPT91,FL05,AABBBDGZ17,LLP20,RPW21,AGHRZ17,ABGHZ18}, in network
communications \cite{SK11} or in data storage \cite{RKSV14}. Compared
to the Hamming world, only few families of codes endowed with the rank
metric are known to have efficient decoding algorithms. Gabidulin
codes are the rank-metric analogue of Reed-Solomon codes and are
somehow optimal because they reach the rank-metric Singleton bound and
come with efficient decoders up to the unique decoding radius
$\frac{n-k}{2}$. However, there exist no known decoder beyond this
bound, even probabilistic ones. More, there exist families of
Gabidulin codes that cannot be list decoded in polynomial time
\cite{RW15}. Nonetheless, if we consider $u$ codewords in parallel, it
is possible to overcome this restriction and decode up to
$\frac{u}{u+1}(n-k)$ rank errors with overwhelming probability.

In the present article, we recall a right-hand side decoder for
Gabidulin codes recently introduced in \cite{BC21} for cryptanalytic
applications. While the aforementioned reference restricted to the
case of full length Gabidulin codes ({\em i.e. n=m}), in the present
article we extend it to handle Gabidulin codes of any length
$n \leq m$.

Next, we show how this decoder can be used to provide a simple decoder
for $u$-interleaved Gabidulin codes. We claim that the use of this
algorithm provides a much simpler point of view on the decoding of
interleaved Gabidulin codes because it only involves solving an
overdetermined linear system. In particular, this algorithm is very
similar to the decoder for Interleaved Reed-Solomon codes presented in
\cite{BKY03}, and its decoding radius is given by comparing the number
of equations to the number of unknowns. Moreover, it permits to clarify
a cryptographic attack based on the decoding of interleaved Gabidulin
codes and provides a very simple explanation of the decoding failures.

\section{Notations and Prerequisites}

In this article, $q$ is a prime power and $k, m, n, u$ are non
negative integers. $\Fq$ denotes the finite field with $q$ elements,
and for a non negative integer $\ell$, $\F{q^{\ell}}$ is the algebraic
extension of $\Fq$ of degree $\ell$. The space of $m\times n$ matrices
with entries in a field $\K$ is denoted by $\Mspace{m\times n}{\K}$.
Lower case bold face letters such as $\bfx$ represent vectors, and
upper case bold face letters such as $\bfX$ denote matrices.

\subsection{Rank metric codes}

Given a vector $\bfx\in\Fqm^{n}$, the \textit{column support} (or
\textit{support}) of $\bfx$, denoted $\Supp(\bfx)$ is the $\Fq$-vector subspace
of $\Fqm$ spanned by the entries of $\bfx$:
\[ \Supp(\bfx) = \span_{\Fq}\{x_{1}, \dots, x_{n}\}.
\] We consider also another notion of support, namely the \textit{row}
support.  Let $\BB$ be a basis of the extension field $\Fqm/\Fq$. Then
we define the extension of $\bfx$ with respect to $\BB$ as the matrix
$\ext{\bfx}{\BB}\in\Mspace{m \times n}{\Fq}$ whose columns are the
decompositions of the entries of $\bfx$ in the basis $\BB$.  The row space of
$\ext{\bfx}{\BB}$ is called the \textit{row} support of $\bfx$, \ie
\[ \Rsupp(\bfx) \eqdef \{\bfy\ext{\bfx}{\BB} \mid \bfy \in \Fq^m\} \subset \Fq^{n}.
\]
It is a vector subspace of $\Fq^{n}$. Notice that the above definition
does not depend on the choice of the basis $\BB$. The \textit{rank
  weight} $\rk(\bfx)$ (or \textit{rank}) of $\bfx$ is the rank of
$\ext{\bfx}{\BB}$ with respect to any basis $\BB$. The \textit{rank
  distance} or \textit{distance} between two vectors
$\bfx, \bfy\in\Fqm^{n}$ is defined as
\(d(\bfx, \bfy) = \rk(\bfx-\bfy).
\) In this article, we consider only $\Fqm$-linear codes: a code $\code{C}$ of
length $n$ and dimension $k$ is an $\Fqm$-subspace of $\Fqm^{n}$ whose dimension, as an
$\Fqm$-vector space, is $k$. The {\em minimum distance} of $\code{C}$ is defined
as
\[ d_{min}(\code{C}) = \min\{\rk(\bfc) \mid \bfc \in \code{C}, \bfc\neq 0\}.
\]

\subsection{Gabidulin codes and $q$-polynomials}\label{ss:Gab+qpoly}

A $q$--polynomial is a polynomial of the form
\[
  P(X) = p_{0}X + p_{1}X^{q}+ \dots + p_{r}X^{q^{r}}, \quad p_{r} \neq 0.
\]
The integer $r$ is called the {\em $q$-degree} of $P$ and is denoted by
$\deg_{q}(P)$. A $q$-polynomial $P$ induces an $\Fq$-linear map $\Fqm\to\Fqm$
whose kernel has dimension bounded by $\qdeg (P)$. The {\em rank} of
a $q$-polynomial will denote the rank of the induced linear map. Let $\qpoly$ be
the space of $q$-polynomials and given a positive integer $k<m$, we denote by
$\qpoly_{<k}$ (resp. $\qpoly_{\le k}$) the space of $q$-polynomials of
$q$-degree less than (resp. less than or equal to) $k$. Equipped with the
addition and the composition law, $\qpoly$ is a non commutative ring which is
left and right euclidean \cite{G96b} and the
two-sided ideal $(X^{q^{m}}-X)$ is the kernel of the canonical map
\[ \qpoly \rightarrow \textrm{Hom}_{\Fq}(\Fqm, \Fqm),
\] inducing an isomorphism
\[ \qpoly/(X^{q^{m}}-X) \simeq \textrm{Hom}_{\Fq}(\Fqm, \Fqm).
\] Let $n\le m$, $k\le n$ and $\bfg = (g_{1}, \dots, g_{n})\in\Fqm^{n}$ whose
entries are linearly independent. The {\em Gabidulin code} of dimension $k$ and
evaluation vector $\bfg$ is defined as
\[ \Gab{k}{\bfg} \eqdef \left\{ (P(g_{1}), \dots, P(g_{n}))\mid P\in \qpoly_{<k}\right\} \subset \Fqm^{n}.
\] Notice that the canonical map
\[ \map{\qpoly_{< k}}{\Gab{k}{\bfg}}{P}{(P(g_1), \dots, P(g_n)).}
\] is rank preserving, which allows to identify $\Gab{k}{\bfg}$ with
$\qpoly_{<k}$. It is well known that Gabidulin codes are {\em Maximum Rank
Distance} (MRD), which means that they reach the rank metric analogue of the
{\em Singleton} bound
\[ d_{min}(\qpoly_{<k}) = n-k+1.
\]
Moreover, Gabidulin codes come with efficient decoders able to correct
errors up to the unique decoding radius $\frac{n-k}{2}$. However,
contrary to Reed-Solomon codes, there exists families of Gabidulin
codes that {\em cannot} be list decoded in polynomial time beyond this
bound \cite{RW15}.

Following \cite[\S~1]{BC21}, to any class $P \in \qpoly / (X^{q^m}-X)$,
corresponds an adjoint $P^\vee$ defined as follows:
\[
  \text{for}\quad P(X) = \sum_{i=0}^{m-1} p_iX^{q^i}\qquad \text{and} \qquad
  P^\vee(X) = \sum_{i=0}^{m-1} p^{q^{m-i}}_i X^{q^{m-i}}.
\]
This corresponds to the usual notion of the adjoint endomorphism
with respect to the non degenerate bilinear form on $\Fqm$:
\(
  (x,y) \longmapsto \text{Tr}_{\Fqm/\Fq}(xy).
\)

 \section{Right hand side decoding algorithm}

Let $k<n\le m,~\bfg = (g_{1}, \dots, g_{n})\in\Fqm^{n}$ whose entries
are linearly independent, and let $\CC \eqdef \Gab{k}{\bfg}$. In
\cite{L06a}, Loidreau introduced a Berlekamp-Welch-like decoding
algorithm, which can decode up to $\lfloor\frac{n-k}{2}\rfloor$ rank
errors. This algorithm works {\em on the left} and can be applied to
Gabidulin codes of arbitrary length $n\le m$. Indeed, representing
vectors in $\Fqm^n$ as matrices, the left--hand side decoding consists
in acting on matrices on the left, which is possible whatever the
length $n$ (which corresponds to the number of columns of the
corresponding matrices).

In \cite{BC21}, the authors proposed to work on the right-hand side
instead, which was useful to provide an attack on the code-based
encryption scheme {\sc Ramesses} \cite{LLP20}. However, their decoding
algorithm was only considered in the case where $n=m$ (which was
enough for cryptanalysis). The right-hand side algorithm is not
completely straightforward when $n < m$. In particular, one can no
longer transpose the matrices representing codewords.

In the present section, we recall a self-contained presentation of the
right-hand side version, and prove that restricting $n$ to be maximal is
unnecessary. In particular, we show how the right-hand side decoding algorithm
applied to any $[n, k]$ Gabidulin code can correct up to
$\left\lfloor\frac{n-k}{2}\right\rfloor$ errors.

\subsection{$n=m$}

Suppose we receive a word $\bfy = \bfc + \bfe \in\Fqm^{n}$ where $\bfc\in\CC\eqdef \Gab{k}{\bfg}$
and $\bfe\in\Fqm^{n}$ has rank $t\le \frac{n-k}{2}$. By linear interpolation, there exist three $q$-polynomials $C\in\qpoly_{<k}$
and $Y, E\in\qpoly_{<m}$ such that
\[ Y = C+E,
\] and $Y$ is known to the receiver (See for instance \cite[Chapter
3]{W13}). Since $n=m$, $(g_{1}, \dots, g_{n})$ forms a basis of the
extension field $\Fqm/\Fq$. Therefore, $\rk(E) = \rk(\bfe) = t$. The
core of the algorithm relies in the following proposition:

\begin{proposition}\label{prop:right_annihilator} Let $E$ be a $q$-polynomial of rank $t$. There exists a
unique monic $q$-polynomial $\Lambda$ such that $\deg_{q}(\Lambda)\le t$ and $E\circ \Lambda = 0$
modulo $(X^{q^{m}} - X)$.
\end{proposition}

The goal is now to compute this right annihilator $\Lambda$. It satisfies the equation
\[ Y\circ \Lambda = C\circ \Lambda + E\circ \Lambda \equiv C\circ \Lambda\mod (X^{q^{m}}-X),
\] which yields a non linear system of $n$ equations whose unknowns are the
$k+t+1$ coefficients of $C$ and $\Lambda$.

\begin{equation}\label{eq:right_wb_nl_system} \left\lbrace
	\begin{array}{l} (Y\circ \Lambda)(g_{i}) = C\circ \Lambda(g_{i}) \\ \qdeg \Lambda \le t \\ \qdeg C \le k-1.
	\end{array} \right.
\end{equation}

\noindent In order to linearize the system, set $N=C\circ \Lambda$ and consider instead

\begin{equation}\label{eq:right_wb_sl_system} \left\lbrace
	\begin{array}{l} (Y\circ \Lambda)(g_{i}) = N(g_{i}) \\ \qdeg \Lambda \le t \\ \qdeg N \le k + t - 1,
	\end{array} \right.
\end{equation}
whose unknowns are the $k+2t+1$ coefficients of $N$ and $\Lambda$. The relationships
between those two systems are specified in the following proposition.

\begin{proposition}\
  \begin{itemize}[label=$\bullet$]
    \item Any solution $(\Lambda, C)$ of \eqref{eq:right_wb_nl_system} gives a
solution $(\Lambda, N=C\circ \Lambda)$ of \eqref{eq:right_wb_sl_system}.
    \item Assume that $E$ is of rank $t\le \lfloor \frac{n-k}{2}\rfloor$. If
$(\Lambda, N)$ is a nonzero solution of \eqref{eq:right_wb_sl_system} then
$N = C\circ \Lambda$ where $C = Y-E$ is the interpolating $q$--polynomial of the
codeword.
  \end{itemize}
\end{proposition}

Therefore, it is possible to recover the codeword $C$ from any non zero solution
$(\Lambda, N)$ of \eqref{eq:right_wb_sl_system} by computing the right hand side
euclidean division of $N$ by $\Lambda$ which can be done efficiently.

\begin{remark}
  Actually, the previous system is only linear over $\Fq$, not over
  $\Fqm$. To address this issue, one can use the adjunction
  operation. Details can be found in \cite{BC21}.
\end{remark}

\subsection{$n<m$}\label{ss:n<m}

Assume now that $n$ is not maximal, and consider a received word
$\bfy = \bfc + \bfe$, where $\bfc = (C(g_1), \dots, C(g_n))$ for some
$q$--polynomial $C$ of $q$--degree $<k$ and $\bfe \in \Fqm^n$ has rank $t$ whose
value is discussed further.

As in the previous section, our first objective is to reformulate the
decoding problem in terms of $q$--polynomials instead of vectors. Here
lies the first issue. Indeed, since $\bfy$ has length $n < m$ there is
not a unique $q$--polynomial $Y$ in $\qpoly/(X^{q^m}-X)$ such that
$\bfy = (Y(g_1),\dots, Y(g_n))$. For this reason, when choosing such an
arbitrary interpolator $Y$ for $\bfy$, one can define $E \eqdef Y -C$ and get
a new $q$--polynomial formulation of the decoding problem as
\[
  Y = C + E,
\]
but here there is no longer any reason that $E$ would have rank $t$,
we only know that the vector $(E(g_1), \dots, E(g_n))$ has rank $t$.
In terms of linear operators, this means that the restriction
of $E$ to the span $V$ of $g_1, \dots, g_n$ over $\Fq$ has rank $t$.

To fix this issue we proceed as follows. First we choose $Y$ as the
interpolator of lowest degree by choosing the unique monic
interpolator of degree $< n$. Since $\qdeg (C) < k < n$, this entails
that $\qdeg(E) < n$. Next we will change the interpolating polynomials
$Y$ and $E$ in order $E$ to have rank $t$. This should be done without
knowing the error. We need a slight generalization of
Proposition~\ref{prop:right_annihilator} which we prove here for the
sake of completeness.

\begin{proposition}\label{prop:co-interpolator}
  There exists a $q$--polynomial $G$ of $q$--degree $\leq m-n$ whose
  image equals the $\Fq$--span of $g_1, \dots, g_n$.
\end{proposition}

\begin{proof}
  Let $V$ denote the $\Fq$--span of $g_1, \dots, g_n$. By
  interpolation, it is well-known that there exists $G_0$ of
  $q$--degree $\leq m-n$ whose kernel equals the $(m-n)$--dimensional
  space $V^\perp$ for the inner product
  $(x,y) \mapsto \text{Tr}_{\Fqm/\Fq}(xy)$. Then, the $q$--polynomial
  $G_1 \eqdef X^{q^{n}}\circ G_0$ has the same kernel $V$ and is in
  the span of the $q$--monomials $X^{q^{n}},\dots, X^{q^{m}}$. Then,
  let $G \eqdef G_1^{\vee}$ be the adjoint as introduced in
  \S~\ref{ss:Gab+qpoly}. It has $q$--degree $\leq m-n$ and, by
  adjunction properties, satisfies
  $\text{Im}(G) = \text{Im}(G_1^\vee) = \ker(G_1)^\perp = V$.
\end{proof}

Let $G$ be the $q$--polynomial of degree $\leq m-n$ given by
Proposition~\ref{prop:co-interpolator} and consider $Y \circ G$ instead of $Y$,
then we get a new problem which is
\begin{equation}\label{eq:new_decoding_pb}
  Y \circ G = C \circ G + E \circ G.
\end{equation}
First, since $\bfy$ and $\bfg$ are known, the $q$--polynomials $Y, G$
are computable using interpolation.  Then, the $q$--polynomial
$C \circ G$ has $q$--degree $< k + m - n$ and hence corresponds to a
codeword of a Gabidulin code of dimension $k+m-n$.  Finally,
$E \circ G$ has rank $t$. Indeed, as mentioned earlier, $t$ is the
rank of the restriction of $E$ to the span of $g_1, \dots, g_n$, which
is precisely the image of $G$. Thus, the reformulated problem
(\ref{eq:new_decoding_pb}) can be regarded as correcting a rank $t$
error in a Gabidulin code of length $m$ and dimension $k + m - n$.
Using our right-hand decoding algorithm it is hence possible to
correct an amount of errors up to
\[
  t = \frac{m-(k+m-n)}{2} = \frac{n-k}{2}\cdot
\]

\begin{remark}
  The previous results may be interpreted in terms of decoding a
  length $m$ Gabidulin code which was column--punctured on the right
  at $\delta = m-n$ positions (see \cite[\S~2.3]{S19a} for a
  definition of column--puncturing).  Similarly, this can be
  unterstood in terms of decoding a length $m$ Gabidulin code under
  $\delta$ rank erasures and $t$ rank errors. In this situation we recover
  the usual fact that $2t + \delta \leq n-k$.
\end{remark}

 \section{Decoding interleaved Gabidulin codes}

\subsection{Interleaving}

Interleaving a code $\code{C}$ consists in considering several
codewords of $\code{C}$ {\em at the same time}, corrupted by errors
having the same support. In the Hamming metric, interleaved
Reed-Solomon codes have been extensively studied and come with
efficient {\em probabilistic} decoders allowing to correct {\em
  uniquely} almost all error patterns slightly beyond the error
capability of the code. See \cite{BKY03} for further reference. In the
rank metric, interleaved Gabidulin codes have been introduced by
Loidreau and Overbeck in \cite{LO06}. Let $\bfg$ be an evaluation
vector, and let $u \in \N$. The $u$-interleaved Gabidulin code of
evaluation vector $\bfg$ and dimension $k$ is

\[
  \IGab{u}{k}{\bfg} \eqdef \left\lbrace
    \begin{pmatrix}
      \mathbf{c^{(1)}} \\  \vdots \\ \mathbf{c^{(u)}}
    \end{pmatrix} \mid \mathbf{c^{(i)}}\in Gab_{k}(\mathbf{g})\right\rbrace.
\]

\begin{remark}
For $u=1$, we recover usual Gabidulin codes.
\end{remark}

\noindent Each codeword $\bfC \in \IGab{u}{k}{\bfg}$ is the evaluation of a {\em column} vector of $q$-polynomials of bounded $q$-degrees on $\bfg$:

\[
  \bfC = (\bfGamma(\bfg_{1}), \dots, \bfGamma(\bfg_{n})), \quad \bfGamma = \begin{pmatrix}C^{(1)} \\ \vdots\\ C^{(u)}\end{pmatrix} \text{ where } C^{(i)}\in\qpoly_{<k}.
\]

\noindent Using the inverse extension map, each $\bfGamma(g_{i})$ can
be interpreted as an element of $\Fqmu$, and $\IGab{u}{k}{\bfg}$ is
then a code of length $n$ and dimension $k$ over $\Fqmu$. Moreover,
they are known to be MRD (see \cite[Lemma 2.17]{W13}), so the error
correction capability of $\IGab{u}{k}{\bfg}$ is $\frac{n-k}{2}$.
However, their specific structure allows to design efficient
algorithms being able to {\em uniquely} decode $\IGab{u}{k}{\bfg}$ up
to $\frac{u}{u+1}(n-k) > \frac{n-k}{2}$ for $u>1$, with very high
probability \cite{LO06,WZ14}.

In this section we show how to use the right-hand side variant of the
Berlekamp-Welch algorithm introduced before to decode $u$-interleaved
Gabidulin codes, up to $\frac{u}{u+1}(n-k)$. Since this is beyond the
error capability of the code, this algorithm might fail but the
probability of failure is very low.

\subsection{Error model}

Similarly to the Hamming metric, we consider a channel model where
errors happen {\em in burst}. In this model, the transmitted codeword
is a matrix $\bfC \in \Mspace{u\times n}{\Fqm}$ representing $u$ codewords
of $\Gab{k}{\bfg}$ {\em in parallel}, and the error is a matrix
$\bfE\in\Mspace{u \times n}{\Fqm}$ of $\Fq$-rank $t$, \ie such that the
matrix of $\Mspace{um \times n}{\Fq}$ obtained from $\bfE$ by extending
every {\em row} in a basis of $\Fqm/\Fq$ is of rank $t$. The receiver
then gets a word $\bfY = \bfC + \bfE$, and the goal is to recover
$\bfC$.

\noindent In the Hamming metric, the receiver gets $u$ noisy
$\bfy^{(i)} = \bfc^{(i)}+\bfe^{(i)}$ such that $\bfc^{(i)}$ are
codewords of some code $\code{C}$ (\eg a Reed-Solomon code) and {\em
  all} the $\bfe^{(i)}$ have the {\em same} support of cardinality
$t$.

\noindent In the current setting, each row of $\bfY$ is of the form
$\bfy^{(i)} = \bfc^{(i)} + \bfe^{(i)}$ where
$\bfc^{(i)}\in\Gab{k}{\bfg}$. The following proposition whose proof is
straightforward justifies the term {\em burst rank-errors}.

\begin{proposition}\label{prop:common_row_support}
  The row support of each $\bfe^{(i)}$ is contained in the $\Fq$-row
  space of $E$ which is of dimension $t$.
\end{proposition}

\begin{remark}
  In this article the error model consists in considering error
  vectors $\bfe^{(i)}$ sharing the same {\em row support}. It seems to
  be the most natural error model when considering the code regarded
  as a code over $\Fqmu$, and it is the one used in most references.
  On the other hand, one may consider another error model where the
  errors share a common {\em column support}. In the latter case, the
  usual left--hand side decoder can be used, see for instance
  \cite{LO06}.
\end{remark}

\subsection{Right-hand side decoding of interleaved Gabidulin codes}\label{subs:right_decoding_interleaved}

Let $\bfY = \bfC + \bfE \in \Mspace{u \times n}{\Fqm}$ be a received word.
By linear interpolation of each row of $\bfY$, there exist $u$ triple
of $q$-polynomials $(Y_{i}, C_{i}, E_{i})$ such that
\[
  Y_{i} = C_{i} + E_{i},
\]
and $\qdeg(C_{i})<k$. Since the errors have the same support of
dimension $t$, there exists a $q$-polynomial $\Lambda$ with
$\qdeg(\Lambda) \le t$ that locates {\em all} the errors. More
specifically, Proposition \ref{prop:common_row_support} induces the
following lemma:

\begin{lemma}\label{lem:common_locator_polynomial}
  Denoting by $E_{i}$ the interpolator $q$-polynomial of $\bfe^{(i)}$,
  there exists $\Lambda \in \qpoly_{\le t}$ such that
  \[
    E_{i}\circ \Lambda = 0 \mod (X^{q^{m}}-X),\quad \forall i\in\{1,\ldots,u\}.
  \]
\end{lemma}

\noindent Lemma \ref{lem:common_locator_polynomial} yields the
following non-linear system of $u\times n$ equations

\begin{equation}\label{eq:interleaved_right_wb_nl_system} \left\lbrace
    \begin{array}{l}
      (Y_{i}\circ \Lambda)(g_{j}) = (C_{i}\circ \Lambda)(g_{j}) \\
      \qdeg \Lambda \le t \\
      \qdeg C_{i} \le k-1, \quad \text{ for  } i\in \{1,\dots, u\}.
	\end{array} \right.
\end{equation}
which can be linearized into the following system, setting
$N_{i} \eqdef C_{i}\circ \Lambda$:

\begin{equation}\label{eq:interleaved_right_wb_system} \left\lbrace
    \begin{array}{l}
      (Y_{i}\circ \Lambda)(g_{j}) = N_{i}(g_{j}) \\
      \qdeg \Lambda \le t \\
      \qdeg N_{i} \le k+t-1, \quad \text{ for  } i\in \{1,\dots, u\}.
	\end{array} \right.
\end{equation}
This system has $u\times n$ equations, and $t+1+u(k+t)$ unknowns, and
therefore one can expect to retrieve $(\Lambda, N_{1},\ldots, N_{u})$
whenever $t\le \frac{u}{u+1}(n-k)$. Since
$N_{i} = C_{i}\circ \Lambda$, the codewords $C_{1},\dots, C_{u}$ can
then be recovered by computing euclidean division on the right.

\begin{remark}
  The decoding algorithm mentioned in \cite{LO06,WZ14} can be
  re-interpreted in terms of the aforementioned decoder. The present
  section permits in particular to shed light on the fact that
  previous algorithms are actually very comparable to Loidreau's
  original algorithm when acting on the right instead of acting on the
  left.
\end{remark}
\subsection{Application to cryptography: \LIGA{} encryption scheme}

Let $\Fqm, \Fqmu$ be two algebraic extensions of the finite field
$\Fq$. In \cite{FL05}, Faure and Loidreau introduced a rank metric
encryption scheme with small key size. The originality of the
cryptosystem was to base the security on the hardness of decoding a
(public) Gabidulin code beyond the unique decoding radius. Indeed, the
public key was of the form $\bfk_{pub} = \bfx\bfG + \bfz$ where $\bfG$
is a generator matrix of a public $[n, k]$ Gabidulin code over $\Fqm$
and $\bfx\in\Fqmu^{k}$, $\bfz\in\Fqmu^{n}$ together with
$t\eqdef \rk(z) >  \frac{n-k}{2}$ form the secret key.

However, it was shown in \cite{GOT18} that an attacker could easily
compute $u$ noisy codewords of the Gabidulin code generated by $\bfG$
using the $\Fqm$--linearity of the trace map $\text{Tr}_{\Fqmu/\Fqm}$,
and then recover the secret providing that $t \le \frac{u}{u+1}(n-k)$
(which was always the case to resist other attacks). This really
amounts to decoding the public key with a decoder of $u$--interleaved
Gabidulin codes. In order to repair the scheme, the authors of \LIGA{} proposed instead
to base the security on the hardness of decoding $u$--interleaved
Gabidulin codes. Indeed, they proved that by reducing the rank of
$\bfz$ over $\Fqm$ (while keeping its rank weight $t$ over $\Fq$
higher than the unique decoding radius), it was no longer possible to
recover the secret key. More precisely, denoting by
$\zeta \eqdef \rk[q^{m}](\bfz)$ this rank, they proved by a careful
analysis of known interleaved decoders that a condition for making the
decoder to fail was $\zeta < \frac{t}{n-k-t}$. In particular, in \LIGA{}
they proposed to set $\zeta = 2$.

Using our decoder, we propose a new interpretation of this condition.
Indeed, let $\bfY = \bfC + \bfE \in \Mspace{u\times n}{\Fqm}$ be a
noisy codeword of an $u$-interleaved Gabidulin code. The results of
Section \ref{subs:right_decoding_interleaved} can be strengthen as
follows: If some rows of $\bfE$ share a linear dependency, then the
equations in system \eqref{eq:interleaved_right_wb_system} are no
longer independent. In particular, if $\zeta \le u$ denotes the rank
of $\bfE$ over $\Fqm$, one can refine the reasoning and deduce an
equivalent linear system with $\zeta \times n$ independent equations
for $t+1+\zeta(k+t)$ unkowns.
Therefore, when $t>\frac{\zeta}{\zeta+1}(n-k)$, there are more
unknowns than equations and the decoder fails. This inequality is
exactly the condition $\zeta < \frac{t}{n-k-t}$ from \LIGA{}.

 \section*{Conclusion}
We presented a full version of a right-hand side decoding algorithm
for Gabidulin codes. This algorithm is close to a verbatim translation
of its well--known left--hand counterpart. However, compared to its
left--hand counterpart, it was unclear how to apply it to non
full length Gabidulin codes. This issue has been addressed in the
present article. Moreover, we claim that this algorithm is of
interest for various applications. First, it provides a very natural
approach for the decoding of interleaved Gabidulin codes. It is
actually very comparable to the algorithm proposed by Loidreau and
Overbeck \cite{LO06} but the strong connection with a Berlekamp--Welch
like decoder was not that clear in the aforementioned
reference. Second, this right-hand side decoder already appeared to
provide an interesting tool for cryptanalytic applications.

\bibliographystyle{splncs04}

\end{document}